\documentclass[pra,twocolumn,groupedaddress,superscriptaddress,amsmath,amssymb,a4paper]{revtex4}
\usepackage{geometry}
\usepackage{graphicx}
\usepackage{mathrsfs}
\usepackage{amssymb}
\usepackage{amsmath}
\usepackage{amsthm}
\usepackage{amsbsy}
\usepackage{bm}
\usepackage{hyperref}
\usepackage[T1]{fontenc}
\usepackage{color}
\newtheorem{proposition}{Proposition}
\newtheorem{theorem}{Theorem}

\theoremstyle{definition}

\newtheorem{example}{Example}
\newtheorem{remark}{Remark}

\begin{document}

\title{Time deformations of master equations}

\author{Sergey N. Filippov}

\affiliation{Moscow Institute of Physics and Technology,
Institutskii Per. 9, Dolgoprudny, Moscow Region 141700, Russia}

\affiliation{Institute of Physics and Technology of the Russian
Academy of Sciences, Nakhimovskii Pr. 34, Moscow 117218, Russia}

\author{Dariusz Chru\'sci\'nski}

\affiliation{Institute of Physics, Faculty of Physics, Astronomy
and Informatics \\  Nicolaus Copernicus University, Grudzi{a}dzka
5/7, 87--100 Torun, Poland}

\begin{abstract}
Convolutionless and convolution master equations are the two
mostly used physical descriptions of open quantum systems
dynamics. We subject these equations to time deformations: local
dilations and contractions of time scale. We prove that the
convolutionless equation remains legitimate under any time
deformation (results in a completely positive dynamical map) if
and only if the original dynamics is completely positive
divisible. Similarly, for a specific class of convolution master
equations we show that uniform time dilations preserve positivity
of the deformed map if the original map is positive divisible.
These results allow witnessing different types of non-Markovian
behavior: the absence of complete positivity for a deformed
convolutionless master equation clearly indicates that the
original dynamics is at least weakly non-Markovian; the absence of
positivity for a class of time-dilated convolution master
equations is a witness of essentially non-Markovian original
dynamics.
\end{abstract}

\maketitle

\section{Introduction}

A physical quantum system is never isolated in practice, which
leads to a concept of an open quantum system. The state of such a
system is described by a density operator $\varrho$ on some
Hilbert space $\mathcal{H}$ (positive semidefinite operator with
unit trace). Time evolution of the open system is governed by the
total Hamiltonian $H$ of ``system + environment'' and the initial
state of the environment $\Omega_E$. If the system and environment
are initially factorized, i.e., their state is $\varrho \otimes
\Omega_E$, then the system dynamics is defined by the standard
reduction
\begin{equation}\label{R}
  \varrho(t) = {\rm Tr}_E \left\{ e^{-iHt} \varrho \otimes \Omega_E e^{iHt}
  \right\}.
\end{equation}

\noindent Formula~\eqref{R} defines a dynamical map $\Phi(t)[X] =
{\rm Tr}_E \left\{ e^{-iHt} X \otimes \Omega_E e^{iHt}
  \right\}$, which has an important property
of being completely positive  (CP) and trace-preserving. Complete
positivity means that $\Phi(t) \otimes {\rm Id}_k$ maps any
(possibly entangled) density operator of the system +
$k$-dimensional ancilla into a legitimate density operator.

Physical environments usually have enormously many degrees of
freedom, which makes the dynamics $\varrho(t)$ intractable via
formula~\eqref{R} unless suitable approximations are
made~\cite{alicki-1987,breuer-petruccione-2002,gardiner-2004,carmichael-2009,li-2018}.
Microscopic derivations of system evolution with the help of
projection operator techniques result in either a convolutionless
master equation~\cite{breuer-petruccione-2002}
\begin{equation}
\label{master-equation} \frac{d \varrho(t)}{d t} = L(t)
[\varrho(t)]
\end{equation}

\noindent with a time-local generator
$L(t) : \mathcal{B}(\mathcal{H}) \mapsto \mathcal{B}(\mathcal{H})$ or
a convolution master
equation~\cite{breuer-petruccione-2002,nakajima-1958,zwanzig-1960,haake-1973}
\begin{equation} \label{memory-kernel}
\frac{d\varrho(t)}{dt} = \int_0^t \mathcal{K}(t,t')[\varrho(t')]
dt'
\end{equation}

\noindent with a memory kernel
$\mathcal{K}(t,t'):\mathcal{B}(\mathcal{H})\mapsto\mathcal{B}(\mathcal{H})$.

Only some sufficient conditions on the time-local generator $L(t)$
and memory kernel $\mathcal{K}(t,t')$ are known, which guarantee
complete positivity and trace preservation of the corresponding
dynamical map
$\Phi(t)$~\cite{hall-2008,barnett-2001,budini-2004,shabani-2005,maniscalco-2005,maniscalco-pertuccione-2006,breuer-vacchini-2008,nalezyty-2015,chruscinski-kossakowski-2016,vacchini-2016,chruscinski-kossakowski-2017,zanardi-2017}.

Suppose that master equations~\eqref{master-equation}
and~\eqref{memory-kernel} define a legitimate quantum dynamics,
i.e., a completely positive and trace-preserving dynamical map
$\Phi(t)$. From the quantum information science perspective, the
evolution process $\Phi(t)$ can have peculiar divisibility
properties. If the dynamical map $\Phi(t)$ can be represented in
the form of concatenation $\Phi(t_2) = V(t_2,t_1) \Phi(t_1)$ with
CP intermediate map $V(t_2,t_1)$ for all $t_2
> t_1 \geqslant 0$, then the process $\Phi(t)$ is called CP
divisible. Analogously, if $V(t_2,t_1)$ is positive (P) for all
$t_2 > t_1 \geqslant 0$, then the process $\Phi(t)$ is called P
divisible. P indivisible processes are also referred to as
essentially non-Markovian, whereas CP indivisible but P divisible
processes are sometimes called weakly
non-Markovian~\cite{chruscinski-maniscalco-2014}. CP divisibility
and P divisibility are only two approaches to define Markovian
quantum processes~\cite{wolf-prl-2008,rivas-2010}, many other
approaches include decreasing distinguishability of system
states~\cite{breuer-2009,laine-2010}, monotonicity of quantum
mutual information~\cite{luo-2012}, decreasing capacity of quantum
channels~\cite{bylicka-2014}, independence of evolution with
respect to events preceding the causal break when the system's
state is actively reset~\cite{pollock-2018}, and others. The
reviews of the current status in the discussion of quantum
non-Markovianity are given in the
papers~\cite{rivas-2014,breuer-2016,de-vega-2017,li-2018}.

The goal of this paper is to relate divisibility properties of
$\Phi(t)$ and the behavior of master equations under time
deformations. By time deformation of a master equation we
understand the transformation
\begin{equation}
\label{replacement} \varrho(t) \rightarrow
\widetilde{\varrho}(\tau), \qquad dt \rightarrow d\tau,
\end{equation}

\noindent where
\begin{equation}
\label{time-deformation} \tau(t) = \int_0^t \alpha(t') dt', \qquad
\frac{d\tau}{dt} = \alpha(t),
\end{equation}

\noindent and $\alpha(t)$ is non-negative real function
quantifying the local time stretching.

The naive interpretation of~\eqref{replacement} would be the
replacement of $\Phi(t)$ by $\Phi(\tau(t))$ but this is not the
case if the generator $L(t)$ or memory kernel $\mathcal{K}(t,t')$
is time dependent. In fact, a time deformation~\eqref{replacement}
may result in a non-legitimate master equation. Surprisingly,
non-legitimacy of a deformed master equation is closely related
with the divisibility property of the undeformed dynamics. In this
paper, we reveal this relation.

The paper is organized as follows.

In Sec.~\ref{section-convolutionless}, we show that the time
deformation of the convolutionless master
equation~\eqref{master-equation} results in a legitimate dynamical
maps if and only if the original dynamics is CP divisible. In
Sec.~\ref{section-convolution}, we relate legitimacy of time
deformation of convolution master equation~\eqref{memory-kernel}
with P divisibility of the original dynamics. In
Sec.~\ref{section-conclusions}, brief conclusions are given.

\section{Deformation of convolutionless master equations}
\label{section-convolutionless}

Master equation~\eqref{master-equation} formally defines a
dynamical map $\Phi(t) = T_{\leftarrow} \exp \left( \int_0^t L(t')
dt' \right)$, where $T_{\leftarrow}$ is the Dyson time-ordering
operator. The intermediate map $V(t_2,t_1)$ in concatenation
$\Phi(t_2) = V(t_2,t_1) \Phi(t_1)$ reads $V(t_2,t_1) =
T_{\leftarrow} \exp \left( \int_{t_1}^{t_2} L(t') dt' \right)$.

Time deformation of Eq.~\eqref{master-equation} results in a
modified (\textit{inequivalent}) master equation
\begin{equation}
\label{deformed-rho-tau-t}
\frac{d\widetilde{\varrho}(\tau(t))}{d\tau(t)} = L(t)
[\widetilde{\varrho}(\tau(t))],
\end{equation}

\noindent where the density operator
$\widetilde{\varrho}(\tau(t))$ describes evolution in the deformed
time and the original generator $L(t)$ is applied at time moments
$\tau(t)$, see Fig.~\ref{figure1}(a).

In terms of the original time $t$ Eq.~\eqref{deformed-rho-tau-t}
reads
\begin{equation}
\label{master-equation-modified} \frac{d\varrho(t)}{dt} =
\frac{d\tau}{dt} \frac{d\varrho}{d\tau} = \alpha(t) L(t)
[\varrho(t)].
\end{equation}

We will refer to Eq.~\eqref{master-equation-modified} as a time
deformation of the original time-convolutionless master equation
\eqref{master-equation}.

%%%%%%%%%%%%%%%%%%%%%%%%%%%%%%%%%%%%%%%%%%%%%%%%%%%%%%%%%%%%%%%%%%%
\begin{figure}[b]
\includegraphics[width=9cm]{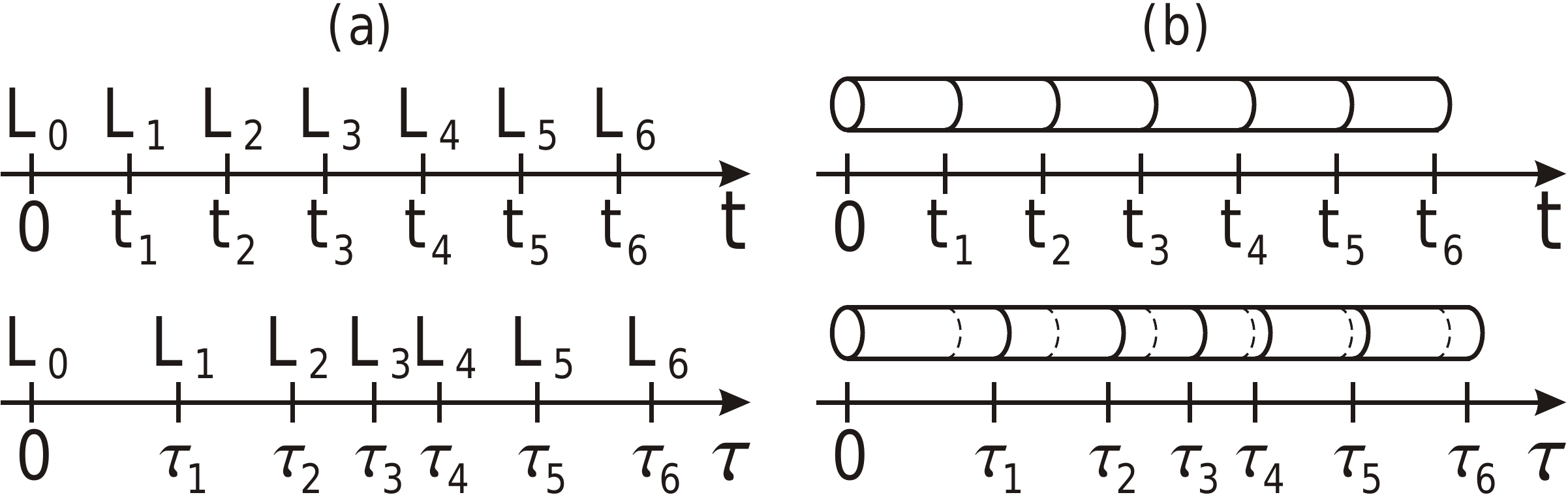}
\caption{\label{figure1} (a) Time deformation of convolutionless
master equation given by time-dependent generator $L(t)$. (b) Time
deformation of CP divisible dynamics. Intermediate channels remain
valid channels under dilations and contractions of time scale.}
\end{figure}
%%%%%%%%%%%%%%%%%%%%%%%%%%%%%%%%%%%%%%%%%%%%%%%%%%%%%%%%%%%%%%%%%%%

If $L$ is time independent, i.e. $\Phi(t) = e^{Lt}$ is a
semigroup, then \eqref{master-equation-modified} results in a
deformed map $\widetilde{\Phi}(t) = \Phi(\tau(t))$. However, if
$L(t)$ is time dependent, then $\widetilde{\Phi}(t) \neq
\Phi(\tau(t))$. Moreover, $\widetilde{\Phi}(t)$ can become not CP
even if the original map $\Phi(t)$ is legitimate (CP and trace
preserving), which can be illustrated by the following example.

\begin{example}
\label{example-1} Consider a qubit map $\Phi(t):
\mathcal{B}(\mathcal{H}_2) \mapsto \mathcal{B}(\mathcal{H}_2)$
given by the generator \cite{hall-2014}
\begin{equation}
\label{Lt-example}
L(t)[\varrho] = \frac{1}{2} \sum_{i=1}^3 \gamma_i(t) ( \sigma_i
\varrho \sigma_i - \varrho),
\end{equation}

\noindent where $\sigma_1,\sigma_2,\sigma_3$ is the conventional
set of Pauli operators, $\gamma_1(t)=\gamma_2(t)=1$, and
$\gamma_3(t) = -\tanh(t)$. The map $\Phi_t$ is CP and trace
preserving for all $t \geqslant 0$, so it is a legitimate
dynamical map that can be realized physically, e.g., in the
deterministic collision model~\cite{fpmz-2017}.

It was shown in Ref.~\cite{benatti-2017} that the time-deformed
master equation
\begin{equation}
\frac{d\varrho(t)}{dt} = \alpha L(t) [\varrho(t)]
\end{equation}

\noindent (obtained via constant time stretching $\tau = \alpha
t$) results in a CP map $\widetilde{\Phi}(t)$ if and only if
$\alpha \geqslant 1$. Thus, if the original master equation is
subjected to a uniform time dilation ($0<\alpha<1$), then the map
$\widetilde{\Phi}_\alpha (t)$ is not CP and does not correspond to
any physical evolution (of initially factorized system and
environment).

Note that $\widetilde{\Phi}(t) \neq \Phi(\alpha t)$ because the
decoherence rates $\gamma_k(t)$ are time dependent. \hfill
$\square$
\end{example}

Non-legitimacy of the deformed map $\widetilde{\Phi}(t)$ in the
example above can be attributed to the fact that the master
equation~\eqref{Lt-example} describes so-called eternal
non-Markovian evolution, i.e. CP indivisible dynamical map
$\Phi(t)$, where $V(t_2,t_1)$ is not CP for all $t_2 >
t_1$~\cite{hall-2014,wudarski-2015,megier-2016}. On the other
hand, if the original dynamical map were CP divisible, then all
the decoherence rates would be positive. Time stretching would not
affect positiveness of decoherence rates and $\widetilde{\Phi}(t)$
would still be a valid dynamical map. This leads us to the
following main result.

\begin{theorem}
\label{theorem-1} Master equation~\eqref{master-equation} with
non-singular generator $L(t)$ describes CP divisible dynamics if
and only if the deformed map remains CP under any time
deformation~\eqref{master-equation-modified}.
\end{theorem}

\begin{proof}
Necessity. Suppose the process $\Phi(t)$ is CP divisible and
$L(t)$ is not singular; then the generator $L(t)$ has the
time-dependent Gorini-Kossakowski-Sudarshan-Lindblad
form~\cite{gks-1976,lindblad-1976}
\begin{eqnarray}
&& L(t)[\varrho] = -i [H(t),\varrho] \nonumber\\
&& + \sum_k \gamma_k(t) \left( A_k(t) \varrho A_k^{\dag}(t) -
\frac{1}{2}\{A_k^{\dag}(t) A_k(t) , \varrho \} \right), \qquad
\end{eqnarray}

\noindent where all the rates $\gamma_k(t) \geqslant 0$.
Multiplication of the Hamiltonian $H(t)$ by $\alpha(t)$ preserves
its Hermiticity, and $\alpha(t) \gamma_k(t) \geqslant 0$, so
$\alpha(t)L(t)$ is still a valid generator of the dynamical map
(see, e.g.,~\cite{hall-2008}).

Sufficiency. Let $\alpha(t) = \left\{ \begin{array}{ll}
  0 & \text{if~} 0 \leqslant t < t_1, \\
  1 & \text{if~} t \geqslant t_1, \\
\end{array} \right.$ then the deformed map $\widetilde{\Phi}(t) =
T_{\leftarrow} \exp \left( \int_0^t \alpha(t') L(t') dt' \right) =
\left\{ \begin{array}{ll}
  {\rm Id}, & \text{if~} 0 \leqslant t < t_1, \\
  V(t,t_1) & \text{if~} t \geqslant t_1. \\
\end{array} \right.$ Therefore, if the deformed map
$\widetilde{\Phi}(t)$ remains CP under any deformation, then
$V(t,t_1)$ is CP too for all $t>t_1$, i.e. the original map
$\Phi(t)$ is CP divisible.
\end{proof}

Therefore, CP divisible dynamics preserves the property of being
CP divisible (and consequently CP) under any time deformation; see
Fig.~\ref{figure1}(b). More importantly, if the original dynamical
map is not CP divisible, then this fact can be revealed by a
suitable time deformation under which the deformed map becomes
nonlegitimate.

\begin{remark} Non-singularity of generator $L(t)$ is needed to
guarantee invertibility of $\Phi(t)$. If $\Phi(t)$ is not
invertible, then CP divisibility of $\Phi(t)$ does not require
positivity of rates $\gamma_k(t)$; see
Refs.~\cite{andersson-2007,chruscinski-rivas-stormer-2017}.
However, the generator is not uniquely defined by the dynamical
map $\Phi(t)$ in this case. In particular, if the process is CP
divisible, then there exists a corresponding (possibly singular)
time-local generator with non-negative rates.
Theorem~\ref{theorem-1} holds true for such generators too.
\end{remark}

\section{Deformation of convolution master equations} \label{section-convolution}

In this section, we consider time deformations of the convolution
master equation~\eqref{memory-kernel} and make implications on P
divisibility of the dynamical map $\Phi(t)$.

Continuing the same line of reasoning as before, let us assume
that the same kernel $\mathcal{K}(t,t')$ is applied at deformed
time moments $\tau(t)$ and $\tau(t')$; see Fig.~\ref{figure2}. As
a result, we obtain a time deformation of
Eq.~\eqref{memory-kernel} of the form
\begin{equation}
\frac{d\widetilde{\varrho}(\tau(t))}{d\tau(t)} = \int_0^{\tau(t)}
\mathcal{K}(t,t')[\widetilde{\varrho}(\tau(t'))] d \tau(t'),
\end{equation}

\noindent which in terms of the original time $t$ reads
\begin{equation}
\label{memory-kernel-modified} \frac{d\varrho(t)}{dt} = \int_0^t
\alpha(t) \alpha(t') \mathcal{K} (t,t') [\varrho(t')] dt'.
\end{equation}

%%%%%%%%%%%%%%%%%%%%%%%%%%%%%%%%%%%%%%%%%%%%%%%%%%%%%%%%%%%%%%%%%%%
\begin{figure}[b]
\includegraphics[width=6cm]{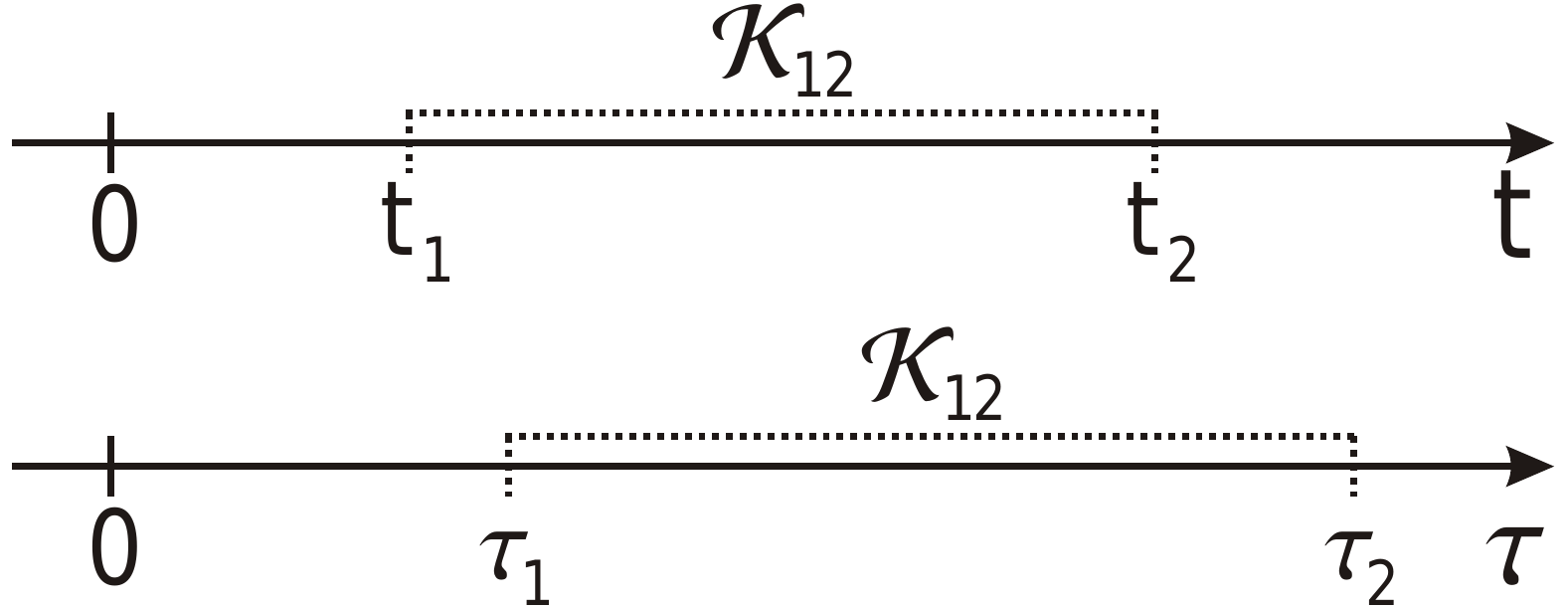}
\caption{\label{figure2} Time deformation of convolution master
equation governed by memory kernel $\mathcal{K}(t,t')$.}
\end{figure}
%%%%%%%%%%%%%%%%%%%%%%%%%%%%%%%%%%%%%%%%%%%%%%%%%%%%%%%%%%%%%%%%%%%

If $\mathcal{K}(t,t') = \delta(t-t')L(t')$, then
\eqref{memory-kernel-modified} reduces to $\frac{d\varrho(t)}{dt}
= \alpha^2(t) L(t)[\varrho(t)]$, i.e., to the time deformation of
the convolutionless master equation considered before.

We assume that the open system dynamics does not depend on the
particular choice of time moment $t=0$, when the system starts
interacting with environment. Due to this time invariance
$\mathcal{K}(t,t') = K(t-t')$~\cite{nakajima-1958,zwanzig-1960}.
In local time deformations~\eqref{time-deformation}, the modified
kernel $\alpha(t) \alpha(t') K (t-t')$ exhibits time invariance
only if $\alpha(t)$ is time independent. For this reason, we
consider only uniform time deformations $\tau(t) = \alpha t$,
$\alpha = {\rm const}$.

Denoting $(A \ast B) (t) = \int_0^t A(t-t')B(t')dt'$, master
equation~\eqref{memory-kernel} takes the form $\frac{d}{dt}
\Phi(t) = (K \ast \Phi) (t)$. Using the Laplace transform $\Phi_s
= \int_0^{\infty} \Phi(t) e^{-st} dt$, the latter equation reduces
to
\begin{equation}
\Phi_s = (s \, {\rm Id} - K_s)^{-1}.
\end{equation}

The uniformly deformed map $\widetilde{\Phi}(t)$ governed by
Eq.~\eqref{memory-kernel-modified} with $\alpha(t)=\alpha$
satisfies
\begin{equation}
\widetilde{\Phi}_s = (s \, {\rm Id} - \alpha^2 K_{s})^{-1}.
\end{equation}

A straightforward algebra yields the following Laplace
transform of the derivative $\frac{d}{dt}\widetilde{\Phi}(t)$:
\begin{eqnarray}
\left( \frac{d\widetilde{\Phi}}{dt} \right)_s &=& \frac{ \alpha^2
\left( \frac{d\Phi}{dt} \right)_{s} }{ {\rm Id} - (\alpha^2
- 1)\left( \frac{d\Phi}{dt} \right)_{s} } \nonumber\\
&=& \alpha^2 \left( \frac{d\Phi}{dt} \right)_{s} \
\sum_{n=0}^{\infty} (\alpha^2 - 1)^n \left[ \left(
\frac{d\Phi}{dt} \right)_{s} \right] ^n, \qquad
\end{eqnarray}

\noindent where the second line represents a valid expansion if
the norm $\| (\alpha^2-1) \left( \frac{d\Phi}{dt} \right)_{s}
\|_{1 \rightarrow 1} < 1$. In the time domain one finds
\begin{eqnarray}
\frac{d\widetilde{\Phi}}{dt} &=& \alpha^2 \frac{d\Phi}{dt} +
\alpha^2 (\alpha^2 - 1)
\frac{d\Phi}{dt} \ast \frac{d\Phi}{dt}  + \ldots \nonumber\\
&& + \alpha^2 (\alpha^2 - 1)^n \underbrace{ \frac{d\Phi}{dt} \ast
\ldots \ast \frac{d\Phi}{dt} }_{n+1\text{~times}} + \ldots
\label{deformed-derivative}
\end{eqnarray}

Let us restrict ourselves to the commutative case, i.e maps
$\Phi(t)$ satisfying $\Phi(t_1)\Phi(t_2) = \Phi(t_2)\Phi(t_1)$ for
all $t_1,t_2 \geqslant 0$. Commutative maps have time-independent
eigenoperators, so the spectrum of $\frac{d\Phi}{dt}$ is merely
the derivative of the spectrum of $\Phi(t)$. Denote eigenvalues of
$\Phi(t)$ by $\lambda_k(t)$, then for P divisible $\Phi(t)$ one
has $\frac{d|\lambda_k(t)|}{dt} \leqslant 0$
\cite{chruscinski-macchiavello-maniscalco-2017}. If $\Phi(t)$ is
Hermitian, i.e., $\Phi(t)$ coincides with its dual map
$\Phi^{\dag}(t)$ in the Heisenberg picture, then $\lambda_k(t)$
are real. Therefore, for commutative Hermitian P divisible maps
$\Phi(t)$ we have $\frac{d\lambda_k(t)}{dt} \leqslant 0$. On the
other hand, if $\frac{d\lambda_k(t)}{dt} \leqslant 0$, then
(\ref{deformed-derivative}) implies
$\frac{d\widetilde{\lambda}_k(t)}{dt} \leqslant 0$ provided $0 <
\alpha < 1$. This way one arrives at the following result.

\begin{proposition}
\label{proposition-kernel} Suppose the commutative Hermitian
dynamical map $\Phi(t)$ is governed by a memory kernel $K(t)$. If the
uniform time dilation $K(t) \rightarrow \alpha^2 K(t)$ with $0 <
\alpha < 1$ and $(1-\alpha^2) \| \left( \frac{d\Phi}{dt}
\right)_{s} \|_{1 \rightarrow 1} < 1$ results in a map
$\widetilde{\Phi}(t)$ such that $\frac{d\widetilde{\Phi}}{dt}$ has
at least one positive eigenvalue at some time $t$, then the
original map $\Phi(t)$ is not P divisible.
\end{proposition}

The class of commutative Hermitian dynamical maps comprises
conventional Pauli qubit maps $\Phi(t)[\varrho] = \frac{1}{2}
\left( {\rm tr}[\varrho] I + \sum_{k=1}^3 \lambda_k(t) {\rm
tr}[\sigma_k \varrho] \sigma_k \right)$ as well as generalized
Pauli
channels~\cite{nathanson-ruskai-2007,chruscinski-siudzinska-2016}.
For Pauli qubit maps one can find a simpler implication of
Proposition~\ref{proposition-kernel}.

\begin{proposition}
\label{proposition-2} Suppose the Pauli map $\Phi(t)$ is governed by
a memory kernel $K(t)$. If the uniform time dilation $K(t)
\rightarrow \alpha^2 K(t)$ with $0 < \alpha < 1$ and $(1-\alpha^2)
(1 - s \int_0^{\infty} \lambda_k(t) e^{- s t} dt) < 1$ results in
a map $\widetilde{\Phi}(t)$, which is not positive, then the
original map $\Phi(t)$ is not P divisible.
\end{proposition}

\begin{proof}
Condition $(1-\alpha^2) (1 - s \int_0^{\infty} \lambda_k(t) e^{- s
t} dt) < 1$ guarantees the validity of
expansion~\eqref{deformed-derivative}. Let $\widetilde{\Phi}(t)$
be non-positive. Since the Pauli map $\widetilde{\Phi}(t)$ is
positive if and only if $-1 \leqslant \widetilde{\lambda}_k(t)
\leqslant 1$, then either $\widetilde{\lambda}_k(t)
> 1$ or $\widetilde{\lambda}_k(t) < - 1$ for some time $t$.
Note that at the initial moment $\lambda_k(0)=1$.

Suppose $\widetilde{\lambda}_k(t)
> 1$; then there exists a time moment $t_0 \in (0,t)$ such that
$\frac{d\widetilde{\lambda}_k(t)}{dt} (t_0) > 0$. By
Proposition~\ref{proposition-kernel}, $\Phi(t)$ is not P
divisible.

Suppose $\widetilde{\lambda}_k(t) < - 1$; let us show that
$\lambda_k(t) \not\geqslant 0$. Using expansion
\begin{eqnarray}
\widetilde{\Phi}(t) &=& \Phi(t) + (\alpha^2 - 1)
(\tfrac{d\Phi}{dt} \ast \Phi) (t) + \ldots \nonumber\\
&& + (\alpha^2 - 1)^n (\underbrace{ \tfrac{d\Phi}{dt} \ast \ldots
\ast \tfrac{d\Phi}{dt} }_{n\text{~times}} \ast \Phi ) (t) + \ldots
,\label{deformed-solution-to-convolution}
\end{eqnarray}

\noindent one finds that if $\lambda_k(t) \geqslant 0$ and
$\frac{d\lambda_k}{dt} \leqslant 0$, then a time deformation with
$0 < \alpha < 1$ guarantees $\widetilde{\lambda}_k(t) \geqslant
0$. As we consider the case $\widetilde{\lambda}_k(t) < - 1$, this
contradiction proves that $\lambda_k(t) \not\geqslant 0$. As a
result, the original Pauli map $\Phi(t)$ is not P divisible.
\end{proof}

The physical meaning of Proposition~\ref{proposition-2} is that
positivity is a topological property of Pauli P divisible process
$\Phi(t)$, which is preserved under uniform time dilations.

\begin{example}
\label{example-2} Consider a pure dephasing qubit map
$\Phi(t)[\varrho] = \frac{1}{2} \left( {\rm tr}[\varrho] I +
\sum_{k=1}^3 \lambda_k(t) {\rm tr}[\sigma_k \varrho] \sigma_k
\right)$ with $\lambda_1(t) = \lambda_2(t) = 1 - 2 \Gamma t e^{-
\Gamma t}$ and $\lambda_3(t) = 1$. This is a valid dynamical map
if $\Gamma > 0$. Such a map is a solution of the convolution
master equation
\begin{eqnarray}
\label{memory-kernel-example} \frac{d\varrho(t)}{dt} & = &
\int_0^t  \left( \Gamma \delta(t-t') - \Gamma^2 \sin \Gamma (t-t')
\right)
\nonumber\\
&& \qquad \times [\sigma_z \varrho(t') \sigma_z - \varrho(t')]
dt'.
\end{eqnarray}

\noindent Condition $(1-\alpha^2) (1 - s \int_0^{\infty}
\lambda_k(t) e^{- s t} dt) < 1$ is fulfilled automatically if $0 <
\alpha < 1$. The uniform time dilation of the memory kernel
$K(t-t') \rightarrow \alpha^2 K(t-t')$ results in the deformed
Pauli map $\widetilde{\Phi}(t)$ with
\begin{equation}
\label{lambda-modified} \widetilde{\lambda}_1(t) =
\widetilde{\lambda}_2(t) = 1 - 2 \alpha^2 e^{- \alpha^2 \Gamma t}
\, \frac{ \sin \left( \sqrt{1-\alpha^4} \, \Gamma t \right)
}{\sqrt{1-\alpha^4}}
\end{equation}

\noindent and $\widetilde{\lambda}_3(t) = 1$. When the
trigonometric function $\sin(\cdot)$ takes negative values,
$\widetilde{\lambda}_1(t) = \widetilde{\lambda}_2(t)
> 1$, see Fig.~\ref{figure3}, so the deformed map
$\widetilde{\Phi}(t)$ is not positive. By
Proposition~\ref{proposition-2}, it clearly indicates that the
original map $\Phi(t)$ is not P divisible.

Note that for the equivalent original \textit{convolutionless}
equation, the uniform time deformation $\tau = \alpha t$ results
in $\widetilde{\lambda}'_i(t) = [\lambda_i(t)]^{\alpha}$,
$i=1,2,3$. In this case, the deformed map $\widetilde{\Phi}'(t)$
remains CP and does not reveal $P$ indivisibility of $\Phi(t)$.
\hfill $\square$
\end{example}

%%%%%%%%%%%%%%%%%%%%%%%%%%%%%%%%%%%%%%%%%%%%%%%%%%%%%%%%%%%%%%%%%%%
\begin{figure}
\includegraphics[width=6cm]{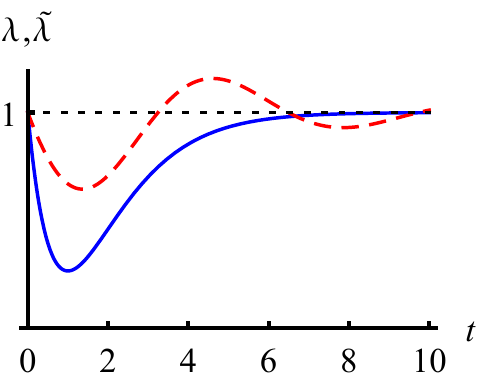}
\caption{\label{figure3} Blue (solid) curve: eigenvalue
$\lambda_1(t)$ of the original dynamical map governed by
convolution master equation~\eqref{memory-kernel-example} with
$\Gamma = 1$. Red (dashed) curve: eigenvalue
$\widetilde{\lambda}_1(t)$ of the time deformed map,
Eq.~\eqref{lambda-modified} with $\Gamma = 1$ and the deformation
coefficient $\alpha = \frac{1}{2}$.}
\end{figure}
%%%%%%%%%%%%%%%%%%%%%%%%%%%%%%%%%%%%%%%%%%%%%%%%%%%%%%%%%%%%%%%%%%%

\begin{example}
\label{example-3} Let us consider a qubit evolution where the
rescaling of the memory kernel is compatible with P divisibility
of the dynamical map. Following~\cite{nalezyty-2015}, let
$\Phi(t)$ be a Pauli qubit dynamical map governed by the memory
kernel
\begin{equation}\label{kernel-example-3}
K(t)[\varrho] = \frac 12 \sum_{k=1}^3 \varkappa_k(t) \sigma_k {\rm
tr}[\sigma_k \varrho],
\end{equation}

\noindent where the time-dependent eigenvalues $\varkappa_k(t)$
are defined (in the Laplace transform domain) via
\begin{equation}\label{kappa}
(\varkappa_k)_s = \frac{-s f_s}{a_k - f_s} .
\end{equation}

\noindent In the above definition the positive numbers
$\{a_1,a_2,a_3\}$ satisfy triangle inequality $a_i^{-1} + a_j^{-1}
\geq a_k^{-1}$ for all permutations of $\{i,j,k\}$, $f(t)$ is a
real function satisfying $f(t) \geqslant 0$ and $f_0 =
\int_0^\infty f(t)dt \leq 4  \left(a_1^{-1} + a_2^{-1} + a_3^{-1}
\right)^{-1}$. The corresponding eigenvalues of $\Phi(t)$ are
given by $\lambda_k(t) = 1 - a_k^{-1} \int_0^t f(t')dt'$.

The dynamical map $\Phi(t)$ is known to be P divisible if
additionally $f(t)$ satisfies the requirement~\cite{nalezyty-2015}
\begin{equation}\label{II}
 f_0 = \int_0^\infty f(t)dt \leq a_{\rm min} ,
\end{equation}

\noindent where $a_{\rm min} = \min\{a_1,a_2,a_3\}$. Suppose
condition~\eqref{II} is fulfilled, then $f_s \leqslant a_{\rm
min}$ for all $s \geqslant 0$. The deformed eigenvalue
\begin{equation}
( \widetilde{\lambda}_k )_s = \frac{1}{s\left( 1 + \frac{\alpha^2
f_s}{a_k - f_s} \right)} = \frac{1}{s}\left( 1 - \frac{f_s}{a_k}
\right) \sum_{n=0}^{\infty} (1-\alpha^2)^n
\left(\frac{f_s}{a_k}\right)^n
\end{equation}

\noindent in time domain is a convolution of two non-negative
functions: the original eigenvalue $\lambda_k(t) \in (0,1]$ and
the inverse Laplace transform of $\sum_{n=0}^{\infty}
(1-\alpha^2)^n \left(\frac{f_s}{a_k}\right)^n$. Hence,
$\widetilde{\lambda}_k(t) \geqslant 0$. If $0 < \alpha < 1$, then
the latter function is less than or equal to the inverse Laplace
transform of $\sum_{n=0}^{\infty} \left(\frac{f_s}{a_k}\right)^n =
\left(1-\frac{f_s}{a_k}\right)^{-1}$. Therefore,
$\widetilde{\lambda}_k(t)$ is less than or equal to the inverse
Laplace transform of function $(\lambda_k)_s
\left(1-\frac{f_s}{a_k}\right)^{-1} = \frac{1}{s}$, i.e.,
$\widetilde{\lambda}_k(t) \leqslant 1$. Thus, the deformed map is
positive for $0 < \alpha < 1$ because the original map is P
divisible; see Eq.~\eqref{II}.

Interestingly, the map $\widetilde{\Phi}(t)$ being positive and
trace-preserving is in general not completely positive and hence
the kernel deformation $K(t) \rightarrow \alpha^2 K(t)$ does not
lead to the legitimate dynamical map. In fact, consider the
behavior of $\widetilde{\lambda}_k(t)$ when $t \rightarrow
\infty$. By the final value theorem
\begin{equation}
\lim_{t \rightarrow \infty} \widetilde{\lambda}_k(t) = \lim_{s
\rightarrow 0} s(\widetilde{\lambda}_k)_s = \frac{1}{1 + \alpha^2
\frac{f_0}{a_k - f_0}}.
\end{equation}

Suppose $a_1 \leqslant a_2 \leqslant a_3$. The deformed map is CP
if and only if the condition $\widetilde{\lambda}_i +
\widetilde{\lambda}_j \leqslant 1 + \widetilde{\lambda}_k$ is
fulfilled for permutations of indices $\{i,j,k\}$. In the limit $t
\rightarrow \infty$ this condition reduces to inequality
\begin{eqnarray}
&& \!\!\!\!\! f_0 \left( \frac{1}{a_2-f_0} + \frac{1}{a_3-f_0} -
\frac{1}{a_1-f_0} \right) + \frac{2 \alpha^2 f_0^2}{(a_2 -
f_0)(a_3 - f_0)} \nonumber\\
&& \!\!\!\!\! + \frac{\alpha^4 f_0^3}{(a_1 - f_0)(a_2 - f_0)(a_3 -
f_0)} \geqslant 0.
\end{eqnarray}

\noindent The obtained inequality is fulfilled for all $0 < \alpha
< 1$ if and only if $(a_2-f_0)^{-1} + (a_3-f_0)^{-1} \geqslant
(a_1-f_0)^{-1}$, which is surprisingly equivalent to CP
divisibility of the original map $\Phi(t)$; cf.
Ref.~\cite{nalezyty-2015}. Thus the dynamical map governed by the
memory kernel~\eqref{kernel-example-3} is CP divisible only if the
deformed map is CP for all $0 < \alpha < 1$. \hfill $\square$
\end{example}

%%%%%%%%%%%%%%%%%%%%%%%%%%%%%%%%%%%%%%%%%%%%%%%%%%%%%%%%%%%%%%%%%%%
\begin{figure}[t]
\includegraphics[width=8.5cm]{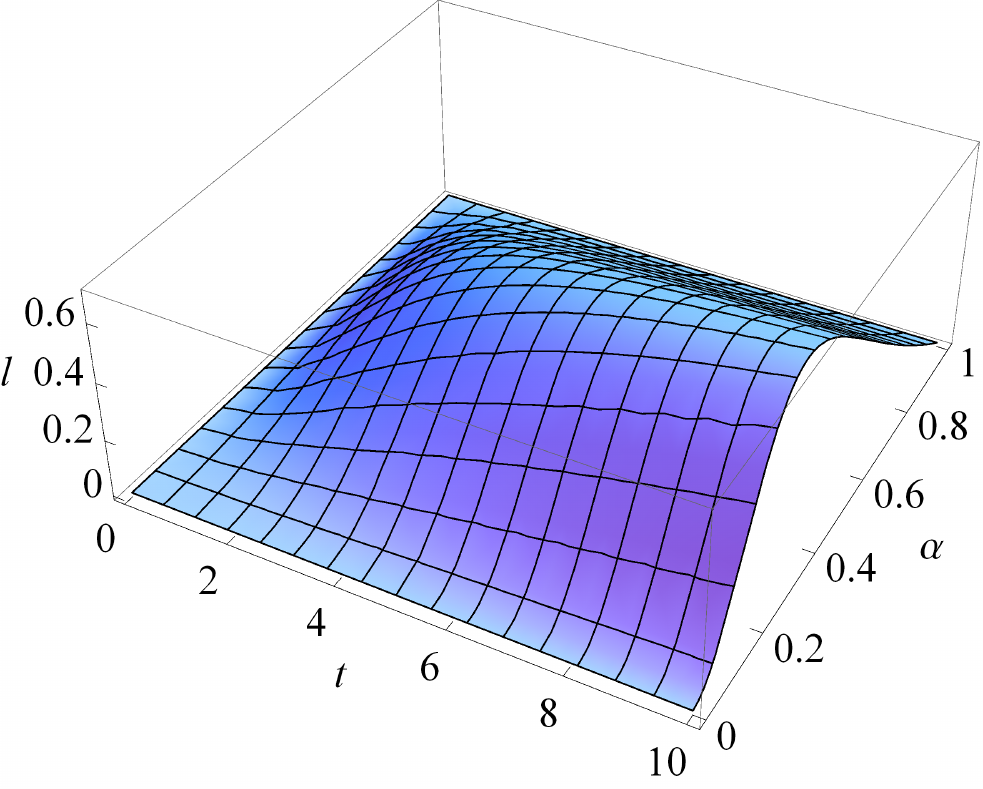}
\caption{\label{figure4} Plot of  $\ell(t) =
\widetilde{\lambda}_1(t) + \widetilde{\lambda}_2(t) -
\widetilde{\lambda}_3(t) - 1$ as a function of dimensionless time
$t$ and rescaling parameter $\alpha$. One has $\ell(t) > 0$ for
all $t>0$ and $0 < \alpha < 1$, so the Fujiwara--Algoet condition
of complete positivity is violated. }
\end{figure}
%%%%%%%%%%%%%%%%%%%%%%%%%%%%%%%%%%%%%%%%%%%%%%%%%%%%%%%%%%%%%%%%%%%

\begin{example} \label{example-4}
Consider CP indivisible Pauli dynamical map $\Phi(t)$ as in
Example~\ref{example-1} but now in terms of the convolution
equation $\frac{d\Phi}{dt} = K \ast \Phi$. The explicit form of
the kernel $K(t)$ is given in Ref.~\cite{megier-2016}. The uniform
time deformation $K(t) \rightarrow \alpha^2 K(t)$ leads to the
deformed eigenvalues $\widetilde{\lambda}_1(t) =
\widetilde{\lambda}_2(t) = (1+\alpha^2)^{-1} [ 1 + \alpha^2
e^{-(1+\alpha^2)t} ]$ and $\widetilde{\lambda}_3(t) = e^{- 2
\alpha^2 t}$. The deformed map $\widetilde{\Phi}(t)$ is never CP
for $t>0$ and $0 < \alpha < 1$ since the corresponding set of
eigenvalues violates the Fujiwara-Algoet conditions for complete
positivity~\cite{fujiwara-algoet-1999} (cf. Fig.~\ref{figure4}).
\hfill $\square$
\end{example}

Considered examples allow us to make a \textit{conjecture} that a
general Pauli dynamical map $\Phi(t)$, defined by a convolution
master equation, is CP divisible if and only if the deformed map
$\widetilde{\Phi}(t)$ is CP for all $0 < \alpha < 1$.  It is
tempting to pose a similar conjecture for general dynamical maps
governed by memory kernel master equations, namely, the map is CP
divisible iff the corresponding rescaled kernel $\alpha^2 K(t)$ is
physically legitimate for  $0 < \alpha < 1$. This, however,
requires further analysis.

\section{Conclusions} \label{section-conclusions}

We have analyzed different forms of non-Markovianity in terms of
the time deformations of governing master equations. If some
deformation of the time-local equation results in a map, which is
not CP, then the original map is not CP divisible (it is at least
weakly non-Markovian). Analogously, if a deformation of the proper
time-convolution equation results is a map, which is not P, then
the original map is not P divisible (it is essentially
non-Markovian).

As the analysis of convolution master equations is particularly
complicated, we have managed to obtain only a necessary condition
for P divisible Hermitian commutative dynamical maps
(Proposition~\ref{proposition-kernel}). We have illustrated
implications of this condition for Pauli dynamical qubit maps
(Proposition~\ref{proposition-2} and Example~\ref{example-2}). We
have also considered Examples~\ref{example-3} and \ref{example-4}
of Pauli dynamical maps defined via a convolution master equation,
for which CP divisibility is equivalent to CP property of the
deformed map for all uniform time dilations.

In addition to witnessing non-Markovianity, the achieved results
clarify legitimate forms of dissipators and memory kernels, which
naturally emerge due to relativistic and gravitational time
dilation~\cite{pikovski-2017} as well as acceleration of quantum
systems~\cite{richter-2017}.

\begin{acknowledgements}

S.N.F. thanks the Russian Science Foundation for support under
project No. 16-11-00084. D.C. was partially supported by the
National Science Centre project 2015/17/B/ST2/02026.

\end{acknowledgements}

\end{document}